\documentclass[12pt]{article}
\usepackage{txfonts}
\usepackage{latexsym}
\usepackage{stmaryrd}
\newtheorem{theorem}{Theorem}
\newtheorem{corollary}{Corollary}
\newtheorem{lemma}{Lemma}

\newtheorem{definition}{Definition}

\newcommand{\blackslug}{\mbox{\hskip 1pt \vrule width 4pt height 8pt 
depth 1.5pt \hskip 1pt}}
\newcommand{\QED}{\quad\blackslug\lower 8.5pt\null\par\noindent}

\newcommand{\eqdef}{\stackrel{\rm def}{=}}

\newcommand{\mysim}{\sim \!}
\newcommand{\dna}{\bindnasrepma}

\title{Non-associative and projective linear logics}
\author{Daniel Lehmann\\
School of Computer Science and Engineering, \\Hebrew University, \\Jerusalem 91904, Israel
\\lehmann@cs.huji.ac.il
}
\date{January 2022}

\begin{document}
\maketitle
\begin{abstract}
A non-commutative, non-associative weakening of Girard's~\cite{Girard:87} linear logic 
is developed for multiplicative and additive connectives. 
Additional assumptions capture the logic of quantic measurements.
\end{abstract}

\section{Introduction} \label{sec:intro}
The novelty of Quantum Mechanics seemed to require a new Logic and such a Quantum Logic
was proposed in~\cite{BirkvonNeu:36} and gave rise to sustained activity on non distributive
lattices. The starting point of this approach is that the atomic propositions of Quantum Logic 
denote (closed) subspaces of a Hilbert space and that the connectives are interpreted as
{\em orthogonal complement}, {\em intersection} and {\em closure of the union}.

In~\cite{Girard:87} Jean-Yves Girard proposed Linear Logic as the logic that could express many
logics and stressed a parallel between some of the features of Linear Logic and some properties 
of quantum systems, expressing his hope that Linear Logic could be more successful in 
explaining the oddities of Quantum Physics. It contains more connectives than the three
connectives just mentioned.
To go {\em quantic} linear logic must go {\em non-commutative} since quantic measurements
are represented by self-adjoint operators that do not always commute.
A non-commutative, but still associative, version of Linear Logic has been developed, 
surprisingly easily, in~\cite{Girard:McGill, Yetter:LL}, but no direct connection to the logics
of quantum measurements has been put in evidence.

This paper claims that the three operations considered by~\cite{BirkvonNeu:36} need to
be complemented by a fourth operation between subspaces already studied 
in~\cite{Lehmann_andthen:JLC}. This operation expresses the temporal composition
of measurements. It is not associative.
This paper develops a non-associative linear logic capable of expressing this operation 
as a connective.

Section~\ref{sec:baby} presents the motivation for the paper: a very limited logic for
describing the possible states of a quantum system after a sequence of measurements
based on a non-associative operation.
Sections~\ref{sec:Q-structures} to~\ref{sec:completeness} present non-associative phase
semantics for the multiplicative and the additive linear connectives, and a sound and complete
set of proof rules. The rules are the rules presented in~\cite{Girard:87} except for the
exchange rule that is replaced by two limited exchange rules.
Section~\ref{sec:projective} presents a restriction on phase semantics, projective structures, 
that validates a limited Weakening rule and seems to fit quantum reasoning.
A sound and complete set of rules for multiplicative and additive connectives 
in projective structures is presented.

\section{A baby quantum logic} \label{sec:baby}
We shall develop a very simple set of quantic propositions and propose that
the connective {\em times} of linear logic express the temporal succession of measurements.
This has a double purpose.
On one hand, it introduces linear logic to those readers interested in the quantum world and
is intended to show them the power of the language of linear logic, before proceeding to its
formal mathematical presentation and, on the other hand,
it is intended to put in evidence before the logicians, both the need for generalizing linear logic 
and the specific properties of quantum logic.

We want to talk about the state of a quantic system, say that certain propositions hold in 
the system, that other propositions do not hold and describe what follows from what.
One gets information about a physical system by performing measurements on it.
The simplest such piece of information is of the type: {\em I measured a certain variable
and I found it has value $x$}.
In Quantum Physics a variable is a self-adjoint operator in some Hilbert space $\cal H$,
a value is an eigenvalue of the operator and finding value $x$ means: {\em the state of 
the system is in the eigen subspace corresponding to the eigenvalue $x$}.
In this first effort we shall assume that any subspace of $\cal H$ can be the eigen subspace 
of some operator.
To justify this assumption we shall assume that the Hilbert space is finite-dimensional
(otherwise we should probably consider only closed subspaces) and we shall assume that no
superselection rules have to be considered.

We shall identify basic propositions with subspaces of $\cal H$. 
Let $\cal P$ be the set of all subspaces of $\cal H$.

The operation ``.'' allows us to describe sequences of propositions, i.e., sequences
of measurements. 
Let \mbox{$A , B \in \cal P$} be subspaces of $\cal H$.
The subspace \mbox{$A . B$} contains the projections on $B$ of the elements of $A$, 
in other terms \mbox{$A . B$} is the projection of $A$ on $B$.
The subspace \mbox{$A . B$} subsumes the proposition:
the system has been measured in subspace $A$ and then measured in subspace $B$.
Note that the operation ``.'', already studied in~\cite{Lehmann_andthen:JLC}, is
neither associative nor commutative.
The space $\cal H$ is a neutral element for ``.'': for any subspace $A$,
\mbox{$ A . {\cal H} =$} \mbox{$ {\cal H} . A = A$}.
The zero-dimensional subspace \mbox{$\{ \vec{0} \}$} is a zero element for ``.'':
\mbox{$ A . \{ \vec{0} \} =$} \mbox{$ \{ \vec{0} \} . A =$} \mbox{$ \{ \vec{0} \} $}.

Another operation is available on subspaces: to any subspace $A$ corresponds its orthogonal
complement $A^\bot$ and, for any subspace $A$, \mbox{$A = {A^\bot}^\bot$}.
The calculus involving ``.'' and $\bot$ has beautiful properties that will be developed in a 
generalization of Girard's linear logic starting in Section~\ref{sec:Q-structures}.


We shall now prove two properties of the operation ``.'' that are crucial for the generalization
of linear logic to be presented.
\begin{lemma} \label{the:baby}
For any subspaces \mbox{$A , B , C $}:
\begin{enumerate}
\item
\mbox{$ A . B = \{ \vec{0} \} $} iff \mbox{$ B . A = \{ \vec{0} \} $},
\item
\mbox{$ ( A . B ) . C = \{ \vec{0} \} $} iff \mbox{$ A . ( C . B ) = \{ \vec{0} \} $}.
\end{enumerate}
\end{lemma}
\begin{proof}
\begin{enumerate}
\item
\mbox{$ A . B = \{ \vec{0} \} $} iff $A$ is orthogonal to $B$ and the orthogonality relation is
symmetric.
\item
Notice, first, that if \mbox{$ x \in A $}, then $x$ is orthogonal to $B$ iff it is orthogonal to
\mbox{$ B . A $}.
Indeed, for any \mbox{$ y \in B $}, one has \mbox{$ y = z + w $} where $z$ is the projection
of $y$ on $A$ and $w$ is the projection of $y$ on $A^\bot$.
Therefore $x$ is orthogonal to $y$ iff $x$ is orthogonal to $z$.

Let \mbox{$ ( A . B ) . C = \{ \vec{0} \} $}, i.e. $A . B$ is orthogonal to $C$.
Assume \mbox{$ x \neq \vec{0} $}, \mbox{$ x \in A . ( C . B ) $}.
Clearly
\begin{itemize}
\item 
$x$ is not orthogonal to $A$ and
\item
\mbox{$ x \in C . B $} and therefore \mbox{$x \in B $} and $x$ is not orthogonal 
to $C$.
\end{itemize}
By the remark above, \mbox{$x \in A . B$} and $x$ is not orthogonal to $C$.
A contradiction.

Suppose now that \mbox{$ x \in A . ( B . C ) $}, \mbox{$x \neq \vec{0} $}.
There are non-null vectors $a$ and $b$ such that \mbox{$a \in A$}, $b$ is the projection
of $a$ onto $B$ and $x$ is the projection of $b$ onto $C$.
By the remark above, since \mbox{$ b \in B$} and $b$ is not orthogonal to $C$, 
$b$ is not orthogonal to $C . B$.
Therefore the projection, say $y$, of $b$ on $C . B$ is not null. 
But $y$ is the projection on $C . B$ of the projection of $a$ onto $B$, and therefore
$y$ is the projection of $a$ onto $C . B$. 
We have shown that \mbox{$A . ( C . B ) $} is not empty.
\end{enumerate}
\QED \end{proof}

\section{Q-structures} \label{sec:Q-structures}
We shall now define structures into which our propositions will be interpreted.
Such structures are a generalization of the phase spaces of Girard's~\cite{Girard:87}.
We shall define multiplicative and additive connectives on such structures, but no exponentials.
We shall provide sound and complete axiomatization for the logic of such structures.

\begin{definition} \label{def:Q-structure}
A Q-structure is a 4-tuple 
\mbox{$ \langle {\cal P}, {\cal Z}, . \ , 1 \rangle $} 
such that
\begin{enumerate}
\item 
$\cal P$ is a set,
\item \mbox{$ {\cal Z} \subseteq {\cal P} $} is a subset of $\cal P$, the {\em garbage} set, 
\item \label{conditions} 
``.'' is a binary operation on $\cal P$ that satisfies, for any \mbox{$x , y , z \in \cal P$},
\begin{enumerate}
\item \label{symmetry}
\mbox{$ x . y \in \cal Z $} iff \mbox{$ y . x \in \cal Z $},
\item \label{reverse}
\mbox{$( x . y ) . z \in \cal Z $} iff \mbox{$x . ( z . y ) \in \cal Z $},
\end{enumerate}
\item
\mbox{$ 1 \in \cal P $} is a neutral element for ``.'', i.e. \mbox{$1 . x =$}
\mbox{$ x . 1 = x $} 
for any \mbox{$x \in \cal P $}.
\end{enumerate}
\end{definition}

Note that 
\begin{itemize}
\item
the operation ``." is not assumed to be associative or commutative,
\item 
condition~\ref{symmetry} is the suitable weakening of the assumption that the 
operation ``.'' is commutative proposed in~\cite{Girard:McGill, Yetter:LL},
\item 
condition~\ref{reverse} is automatically satisfied if the operation ``.'' is both commutative
and associative, but not if it is only associative,
\item
the reason condition~\ref{reverse} has been preferred to the condition 
\mbox{$( x . y ) . z \in \cal Z $} iff \mbox{$x . ( y . z ) \in \cal Z $} is purely circumstantial:
we prefer to consider the information gathered by measurements on the final state and not
the information gathered on the initial state. A structure satisfying this latter condition
instead of condition~\ref{reverse} satisfies mirror images of the properties of Q-structures,
\item
in the presence of~\ref{symmetry} condition~\ref{reverse} is equivalent to:
\mbox{$ ( x . y ) . z \in \cal Z $} iff \mbox{$ ( z . y ) . x \in \cal Z $},
\item
we  use $\cal Z$ where Girard uses $\bot$ because the latter is already heavily overloaded,
\end{itemize}

The garbage set $\cal Z$ and the orthogonality relation, so basic to quantum physics, 
define each other.
\begin{definition} \label{def:orth}
The orhogonality relation on $\cal P$, denoted $\bot$ is defined by:
\mbox{$x \bot y $} iff \mbox{$ x . y \in \cal Z$}.
Then \mbox{$ {\cal Z} = \{ x . y \mid x \bot y \} $}.
\end{definition}

The following is obvious.
\begin{lemma} \label{the:orth}
Condition~\ref{symmetry} is equivalent to the requirement that the relation $\bot$ be
{\em symmetric}.
\end{lemma}

The set $\cal P$ should be understood as the set of all imaginable, possible and imposible, 
situations, possible and impossible worlds.
The set $\cal Z$ is the set of all impossible, contradictory situations.
The operation ``.'' composes two situations.
Two situations are orthogonal iff their composition is impossible.
Condition~\ref{symmetry} requires that orthogonality be symmetric.
The interpretation of Condition~\ref{reverse} is less obvious.

We noticed that any phase space, as defined in~\cite{Girard:87}, is a Q-structure.
Let us describe two more examples of Q-structures.

Our first example is a presentation of the phase semantics of classical propositional logic,
which can throw light on the differences between classical and quantum logic.
Let $V$ be a set of propositional variables, \mbox{$ M = 2^{V}$} the set of models for $V$
and let \mbox{$ {\cal P } =$} \mbox{$ M \cup \{ 0 , 1 \} $}.
Define ``.'' by: for any \mbox{$x , y \in M $} \mbox{$ x . x = x $} and, if \mbox{$x \neq y $},
\mbox{$ x . y = y . x = 0 $}, $1$ is a neutral element and $0$ is a zero
for ``.'': \mbox{$0 . x =$} \mbox{$ x . 0 = 0 $}.
The ``.'' operation is both commutative and associative.

Our second example is a presentation of the baby quantum logic of Section~\ref{sec:baby}
and is to be compared to the previous example.
Given a Hilbert space $\cal H$, the set
$\cal P$ includes all one-dimensional subspaces of $\cal H$, its zero-dimensional subspace 
$\{ \vec{0} \}$ and the space $\cal H$ itself.
Note that not all subspaces of $\cal H$ are elements of $\cal P$.
The set $\cal Z$ is the singleton that contains the zero-dimensional subspace. 
The element $1$ is $\cal H$.
The operation ``.'' is defined by: for any one-dimensional subspaces $x , y$,
\mbox{$ x . y = y $} if $x$ and $y$ are {\em not} orthogonal and
\mbox{$ x . y = \vec{0} $} if \mbox{$ x \bot y $}.
The set $\cal H$ is a neutral element and the $\{ \vec{0} \} $ is a zero element for ``.''.
Lemma~\ref{the:baby} shows that items~\ref{symmetry} and~\ref{reverse} 
of Definition~\ref{def:Q-structure} hold.

Here is a summary of this paper's claims.
\begin{enumerate}
\item 
Q-structures provide the natural extension of Linear Logic to the non-associative,
non-commutative case. Their logic exhibits most of the beautiful symmetries of Linear Logic.
\item 
Associative structures are not fit for Quantum Logic since the basic operation of 
quantum logic is not associative as already noticed in~\cite{Lehmann_andthen:JLC}.
\item
Q-structures in which the garbage set satisfies an additional property are a suitable 
framework for quantum logics as will be shown in Section~\ref{sec:projective}.
\end{enumerate}

\section{Facts} \label{sec:facts}
If the elements of $\cal P$ are the possible situations, the subsets of $\cal P$ represent
the possible states of information about the situation.
In quantum logic, not all subsets of $\cal P$ represent bona fide information states.
For example, an information state that contains \mbox{$x \in \cal P$} and also
\mbox{$y \in \cal P$} must, at least in the absence of a superselection rule, contain all
linear combinations of $x$ and $y$.
This requirement can be formalized in terms of the orthogonality relation.
Girard calls the subsets that represent information states {\em facts} and we shall stick
with his terminology.

\begin{definition} \label{def:orthogonal}
Let \mbox{$A \subseteq \cal P $}.
\begin{equation} 
A^\bot = \{ b \in {\cal P} \mid b \bot a , \ \forall a \in A \}.
\end{equation}
The set $A$ is a {\em fact} iff \mbox{$A = {A^\bot}^\bot$}.
\end{definition}

In the presentation of classical propositional logic of Section~\ref{sec:Q-structures}:
$\cal P$ is a fact and it is the only fact that contains $1$, a subset of $\cal Q$ that does
not contain $1$ is a fact iff it contains $0$.

In baby quantum logic, $\cal P$ is a fact and it is the only fact that contains $1$,
a subset of $\cal Q$ that does not contain $1$ is a fact iff it contains $0$ 
and all the one-dimensional subspaces of a certain subspace.
The facts are in one-to-one correspondence with the subspaces of $\cal H$, as expected.

The following lemma is proved as in~\cite{Girard:87}.
The use of commutativity is replaced by that of Condition~\ref{symmetry} of 
Definition~\ref{def:Q-structure}.
\begin{lemma} \label{the:perp}
For any \mbox{$A , B \subseteq \cal P $}, 
\begin{enumerate}
\item
\mbox{$ A \subseteq {A^\bot}^\bot $},
\item
if \mbox{$B \subseteq A \subseteq \cal P$}, then 
\mbox{${A}^{\bot} \subseteq {B}^{\bot}$},
\item
\mbox{$A^\bot =$} \mbox{$ {{A^\bot}^\bot}^\bot$}.
\end{enumerate}
\end{lemma}

Some basic results about facts will be presented now.
They parallel the material in~\cite{Girard:87}, slightly streamlined and replacing the
commutativity and associativity of ``.'' by the conditions in Definition~\ref{def:Q-structure}.
\begin{lemma} \label{the:fact}
\begin{enumerate}
\item 
A subset \mbox{$F \subseteq \cal P$} is a fact iff there is some 
\mbox{$A \subseteq \cal P$} such that \mbox{$F = A^\bot$}.
\item \label{inter} 
If \mbox{$ \{ F_{i} \} , i \in I$} is a collection of facts, then its intersection 
\mbox{$ \bigcap_{i \in I} F_{i}$} is a fact.
\item 
\mbox{$ {\cal Z} = \{ 1 \}^\bot$}.
Let \mbox{$ {\bf 1} \eqdef {\cal Z}^\bot $}.
$\bf 1$ is a fact, \mbox{$ 1 \in {\bf 1} $}.
If \mbox{$x , y \in {\bf 1} $}, then \mbox{$ x . y \in {\bf 1} $}.
\item
Let \mbox{$ {\bf 0 } \eqdef {\cal P}^\bot $}.
{\bf 0} is the intersection of all facts.
\end{enumerate}
\end{lemma}
\begin{proof}
\begin{enumerate}
\item \label{p}
The {\em only if} part is obvious.
For the {\em if} part, assume \mbox{$F = A^\bot$}.
By Lemma~\ref{the:perp}, we have
\mbox{${F^\bot}^\bot =$}
\mbox{$ {{A^\bot}^\bot}^\bot =$}
\mbox{$A^\bot = F$}.
\item 
By Lemma~\ref{the:perp}, 
\mbox{$ \bigcap_{i \in I } F_{i} \subseteq {{( \bigcap_{i \in I } F_{i} )}^\bot}^\bot $}.
Then, \mbox{$ \bigcap_{i \in I } F_{i} \subseteq F_{j} $} 
for any \mbox{$j \in I $} implies that
\mbox{$  {{( \bigcap_{i \in I } F_{i} )}^\bot}^\bot \subseteq $}
\mbox{$ {F_{j}^\bot}^\bot = F_{j} $} 
for any \mbox{$j \in I $}.
\item 
The first claim is obvious.
By item~\ref{p} above, $\bf 1$ is a fact.
For any \mbox{$ z \in \cal Z $}, \mbox{$ 1 . z = z \in \cal Z $}.
This proves that \mbox{$ 1 \in {\cal Z}^\bot $}.
For any \mbox{$ x , y \in {\bf 1} = {\cal Z}^\bot $} we have,
 for any \mbox{$ z \in {\cal Z } $}, by Definition~\ref{def:Q-structure}, item~\ref{symmetry},
\mbox{$ z . y \in {\cal Z} $}, \mbox{$ x . ( z . y ) \in {\cal Z} $} and therefore, 
by item~\ref{reverse}, \mbox{$ ( x . y ) . z  \in {\cal Z} $} for any such $z$.
We conclude that \mbox{$ x . y \in {\bf 1} $}.
\item
For any \mbox{$A \subseteq \cal Q$}, \mbox{$ {\bf 0} \subseteq A^\bot $}.
Therefore \mbox{$ {\bf 0} \subseteq F $} for any fact $F$.
But {\bf 0} is a fact by Lemma~\ref{the:fact}.
\end{enumerate}
\QED \end{proof}

The operation ``.'' can be applied to subsets of $\cal Q$:
\mbox{$ A . B = \{ x \mid x = a . b , a \in A, b \in B \} $}.
In the presentation of classical logic proposed above, for any facts $F$, $G$,
\mbox{$F . G = F \cap G$}.
In baby quantum logic, since facts are subspaces, for any facts $F$, $G$,
$F . G$ is the projection of $F$ onto $G$.

\section{Multiplicative Connectives} \label{sec:connectives}
We shall now introduce a number of multiplicative connectives.
Connectives transform facts into facts. It is important to remember that the arguments
of a connective must be facts, not arbitrary subsets of $\cal P$ 
and that the result must also be a fact.

\subsection{Linear negation} \label{negation}
Our first connective, {\em linear negation} is unary.
\begin{definition} \label{def:negation}
For any fact $F$, linear negation is defined by \mbox{$ \mysim F = F^\bot $}.
It is a fact by Lemma~\ref{the:fact}.
\end{definition}

\begin{lemma} \label{the:involution}
Linear negation is involutive: for any fact $F$, \mbox{$\sim \mysim F = F$}.
\end{lemma}
\begin{proof}
By Definition~\ref{def:orthogonal}.
\QED \end{proof}

\subsection{The times connective} \label{sec:times}
The multiplicative conjunction, the times connective will be introduced now.
It is denoted $\otimes$.

\begin{definition} \label{def:times}
For any facts $F$, $G$, 
\mbox{$ F \otimes G \eqdef {{( F . G )}^\bot}^\bot $}.
By Lemma~\ref{the:fact}, it is a fact.
\end{definition}

We shall study, now, the properties of the connective {\em times}.
\begin{lemma} \label{the:before_neutral}
In any  Q-structure
\begin{enumerate}
\item
$\bf 1$ is a left-neutral element: 
for any fact $F$ one has \mbox{$ F = $} \mbox{$ {\bf 1} \otimes F $}, and
\item
$\bf 1$ is only half a right-neutral element:
for any fact $F$ one has \mbox{$ F \subseteq F \otimes {\bf 1} $}.
\end{enumerate}
\end{lemma}
The consideration of baby quantum logic shows that \mbox{$ F \otimes I $} is not, in general,
equal to $F$.
\begin{proof}
\begin{enumerate}
\item
Let \mbox{$A \subseteq \cal P $}.
Let \mbox{$x \in A^\bot $}. For any \mbox{$ y \in A $} and any \mbox{$ z \in {\bf 1} =$}
\mbox{${\cal Z}^\bot $}, we have
\mbox{$ x .y \in {\cal Z} $} ,  \mbox{$ ( x . y ) . z \in {\cal Z} $} and
\mbox{$ x . ( z . y ) \in {\cal Z}$}.
We have shown that \mbox{$ A^\bot \subseteq {( {\bf 1} . A )}^\bot $}
and therefore \mbox{$ {{( {\bf 1} . A )}^\bot}^\bot \subseteq {A^\bot}^\bot $}.
For any fact $F$, then \mbox{$ {\bf 1} \otimes F \subseteq F $}.
But, since \mbox{$ 1 \in \bf 1 $}, \mbox{$ F \subseteq {\bf 1} . F $}.
\item
Since \mbox{$ 1 \in {\bf 1} $}, \mbox{$ F \subseteq F . {\bf 1} $} and
\mbox{$ F = {F^\bot}^\bot \subseteq F \otimes {\bf 1} $}.
\end{enumerate}
\QED \end{proof}

The connective $\otimes$ is not associative, but one can show the following.
\begin{lemma} \label{the:product}
In a Q-structure, for any \mbox{$A , B \subseteq {\cal P} $}, one has
\mbox{$ {( A . B )}^\bot \subseteq$} \mbox{$ { ( {A^\bot}^\bot . B )}^\bot $}.
\end{lemma}
\begin{proof}
Suppose \mbox{$x \in { ( A . B )}^\bot$}.
For any \mbox{$a \in A$} and any \mbox{$b \in B$} we have 
\mbox{$( a . b ) . x \in {\cal Z} $}.
Therefore \mbox{$ a . ( x . b ) \in {\cal Z } $} and \mbox{$ x . b \in A^\bot$}.
Consider any \mbox{$c \in {A^\bot}^\bot$}.
We have \mbox{$ ( c . b ) . x = $} \mbox{$c . ( x . b ) \in {\cal Z} $} and we see that 
\mbox{$x \in {( {A^\bot}^\bot . B )}^\bot $}.
\QED \end{proof}

\begin{lemma} \label{the:assoc}
In a Q-structure, for any facts $F$, $G$ and $H$ one has 
\mbox{$ ( F \otimes G ) \otimes H = $} \mbox{$ {( ( F . G ) . H )^\bot}^\bot $}.
\end{lemma}
\begin{proof}
Let $F$, $G$ and $H$ be facts.
By Lemma~\ref{the:product}, \mbox{$ ( ( F . G ) . H )^\bot \subseteq $}
\mbox{$ {( {( F . G )^\bot }^\bot . H )}^\bot $} and therefore, 
by Lemma~\ref{the:perp}, \mbox{$ ( F \otimes G ) \otimes H \subseteq $} 
\mbox{$ {( ( F . G ) . H )^\bot}^\bot $}.
But, \mbox{$ ( F . G ) . H \subseteq $} \mbox{$ ( F \otimes G ) \otimes H $} 
and therefore \mbox{$ ( F \otimes G ) \otimes H = $}
\mbox{$ {{( ( F . G ) . H )}^\bot}^\bot $}.
\QED \end{proof}

\subsection{The parallelization connective} \label{sec:par}

The parallelization connective, denoted $\dna$ and called {\em par} is defined as expected.
\begin{definition} \label{def:par}
\mbox{$ F \dna G =$} \mbox{$ {( F^\bot . G^\bot )}^\bot$}.
\end{definition}
Clearly, \mbox{$ F \dna G $} is a fact.
One easily sees that $\dna$ and $\otimes$ are dual connectives:
\mbox{$ F \dna G =$} \mbox{$ \mysim ( \mysim F \, \otimes \mysim G ) $} and
\mbox{$ F \otimes G =$} \mbox{$ \mysim ( \mysim F \, \dna \mysim G ) $}.
It follows that $\cal Z$ is a left-neutral element for $\dna$: 
\mbox{$ {\cal Z} \dna F = F $}, and half a right-neutral element:
\mbox{$ F \dna {\cal Z} \subseteq F $}.

\subsection{Linear implication} \label{implication}
The linear implication will be denoted by $\multimap$.
\begin{definition} \label{def:implication}
For any facts $F$, $G$, one defines
\mbox{$ F \multimap G = $}
\mbox{$ {( F . G^\bot )}^\bot $}.
\end{definition}
Clearly, \mbox{$ F \multimap G $} is a fact.
Contrary to the commutative case \mbox{$ \mysim G \multimap \mysim F $} is not equal
to \mbox{$ F \multimap G $}.
One sees that, as in the commutative case, 
\mbox{$ F \multimap G = $} \mbox{$ \mysim ( F \, \otimes \mysim G ) = $}
\mbox{$ \mysim F \, \dna \, G $}, \mbox{$ F \otimes G = $} 
\mbox{$ \mysim ( F \multimap \, \mysim G ) $} and \mbox{$ F \dna G = $} 
\mbox{$ \mysim F \, \multimap G $}.

\begin{lemma} \label{the:imp}
For any facts \mbox{$ F , G $}, 
\mbox{$ x \in F \multimap G $} iff \mbox{$ x . h \in F^\bot $} for every 
\mbox{$ h \in G^\bot $}.
\end{lemma}
\begin{proof}
\mbox{$ x \in {( F . G^\bot )}^\bot $} iff \mbox{$ x . ( f . h ) \in {\cal Z} $} for every
\mbox{$ f \in F $} and every \mbox{$ h \in G^\bot $} iff 
\mbox{$ ( x . h ) . f \in {\cal Z} $} for every
\mbox{$ f \in F $} and every \mbox{$ h \in G^\bot $} iff 
\mbox{$ x . h \in F^\bot $} for every \mbox{$ h \in G^\bot $}.
\QED \end{proof}

\begin{corollary} \label{the:1par}
If $F$, $G$ are facts, then \mbox{$ {\bf 1} \subseteq F \dna G $} iff
\mbox{$1 \in F \dna G $} iff \mbox{$ G^\bot \subseteq F $} 
iff \mbox{$  F^\bot \subseteq G $} iff \mbox{$1 \in G \dna F $}.
\end{corollary}
\begin{proof}
\begin{enumerate}
\item
\mbox{$ 1 \in { \{ 1 \}^\bot}^\bot = {\bf 1} $} and $F \dna G$ is a fact.
\item
\mbox{$1 \in F \dna G $} iff \mbox{$ 1 \in \mysim F \multimap G $} iff, 
by Lemma~\ref{the:imp}, \mbox{$ G^\bot \subseteq F $}.
\item
Since $F$ and $G$ are facts.
\item
As above.
\end{enumerate}
\QED \end{proof}

\section{Validity} \label{sec:validity}

\begin{definition} \label{def:valid}
A fact $F$ is said to be {\em valid} in a Q-structure if one of the following, equivalent,
properties hold:
\begin{itemize}
\item
\mbox{$1 \in F$},
\item
\mbox{$ {\bf 1 } \subseteq F $},
\item
\mbox{$ F^\bot \subseteq {\cal Z} $}.
\end{itemize}
\end{definition}
The equivalence of the conditions above is obvious.

The next lemma shows that linear implication expresses deduction.
\begin{lemma} \label{the:implication}
Let $F$, $G$ be facts in a Q-structure: \mbox{$ F \multimap G $} is valid iff 
\mbox{$ F \subseteq G $}. 
\end{lemma}
\begin{proof}
Suppose, first, that \mbox{$ 1 \in F \multimap G $}.
By Lemma~\ref{the:imp}, \mbox{$G^\bot \subseteq F^\bot$} and, therefore,
\mbox{$ F \subseteq G $}.

Suppose, now, that \mbox{$ F \subseteq G $}. 
We have \mbox{$ F . G^\bot \subseteq {\cal Z} $} and therefore 
\mbox{$ 1 \in F \multimap G $} since \mbox{$ {\cal Z} = \{ 1 \}^\bot $}.
\QED \end{proof}

The next lemma shows that the linear negation allows a jump over the turnstile in both
directions.
Note that $G$ jumps from the rightmost position to the rightmost position.
\begin{lemma} \label{the:passing}
In a Q-structure, for any facts $F$, $G$, $H$,
\mbox{$ ( F \otimes G ) \multimap H $} is valid iff  
\mbox{$ F \multimap ( H \, \dna \mysim G ) $} is valid.
\end{lemma}
\begin{proof}
By Lemma~\ref{the:implication} we must show that
\mbox{$ { {( F . G )}^\bot}^\bot \subseteq H $} iff 
\mbox{$ F \subseteq {( H^\bot . G )}^\bot $}.
Assume the former.
We have \mbox{$ H^\bot \subseteq {( F . G )}^\bot $}.
For any \mbox{$ f \in F $}, \mbox{$ g \in G $} and \mbox{$ d \in H^\bot $}, we have
\mbox{$ ( f . g ) . d \in {\cal Z} $} and \mbox{$ f . ( d . g ) \in {\cal Z} $} and we see that
\mbox{$ f \in {( H^\bot . G )}^\bot $}.

Assume, now, that \mbox{$ F \subseteq {( H ^\bot . G )}^\bot $}.
For any \mbox{$ f \in F $}, \mbox{$ g \in G $} and \mbox{$ d \in H^\bot $}, we have
\mbox{$ f . ( d . g ) \in {\cal Z} $} and therefore \mbox{$ f . g \in {H^\bot}^\bot = H $}.
We conclude that \mbox{$ F . G \subseteq H $} and therefore
\mbox{$ { {( F . G )}^\bot}^\bot \subseteq$} \mbox{$ {H^\bot}^\bot = H $}.
\QED \end{proof}

\section{Additive connectives} \label{sec:additives}
\subsection{{\em with}, the additive conjunction} \label{sec:with}
\begin{definition} \label{def:with}
If \mbox{$F, G \subseteq {\cal P} $} are facts, \mbox{$F \& G = F \cap G $}.
\end{definition}
By Lemma~\ref{the:fact} part~\ref{inter}, \mbox{$F \& G $} is indeed a fact.
One sees that {\em with}, i.e. $\&$ is associative, commutative and that
\mbox{$ {\cal P} \& F = F $}.

\begin{lemma} \label{the:dist1}
The connective {\em par } distributes over {\em with}.
For any facts \mbox{$F , G , H $}, \mbox{$ F \dna ( G \& H ) = $}
\mbox{$ ( F \dna G ) \& ( F \dna H ) $} and
\mbox{$ ( G \& H ) \dna F = $} \mbox{$ (G \dna F ) \& ( H \dna f )$}.
\end{lemma}
\begin{proof}
One easily sees that for any \mbox{$A , B \subseteq {\cal P} $}, one has
\mbox{$ {( A \cup B )}^\bot =$} \mbox{$ A^\bot \cap B^\bot $}.
Therefore \mbox{$ ( F \dna G ) \& ( F \dna H ) = $}
\mbox{$ {( ( F^\bot . G^\bot ) \cup ( F^\bot . H^\bot ) )}^\bot = $}
\mbox{$ {( F^\bot . ( G^\bot \cup H^\bot ) )}^\bot $}.
But now, \mbox{$ {( G^\bot \cup H^\bot )}^\bot = $} 
\mbox{$ {G^\bot}^\bot \cap {H^\bot}^\bot =$} \mbox{$G \cap H $} and therefore
\mbox{$ {( G^\bot \cup H^\bot )} = $} \mbox{$ {( G \cap H )}^\bot$} and
\mbox{$  ( F \dna G ) \& ( F \dna H ) = $} 
\mbox{$ {( F^\bot . {( G \cap H )}^\bot )}^\bot = $}
\mbox{$ F \dna ( G \& H ) $}.
The second claim is proved similarly.
\QED \end{proof}

\begin{lemma} \label{the:zero}
Let us define \mbox{$ {\bf 0} = {\cal P}^\bot $}.
${\bf 0}$ is a fact.
For any fact $F$, \mbox{$ {\bf 0} \& F =$} \mbox{$ F \& {\bf 0} = {\bf 0} $}.
\end{lemma}
\begin{proof}
By Lemma~\ref{the:fact}, and since \mbox{$F^\bot \subseteq {\cal P} $} we have
\mbox{$ {\bf 0} \subseteq$} \mbox{$ {F^\bot}^\bot = F$}.
\QED \end{proof}

\begin{lemma} \label{the:semi-dist}
The connective {\em times} semi-distributes over {\em with}, i.e., :
\mbox{$ F \otimes ( G \& H ) \subseteq $} \mbox{$ ( F \otimes G ) \& ( F \otimes H ) $} and
\mbox{$ ( G \& H ) \otimes F \subseteq $} \mbox{$ ( G \otimes F ) \& ( H \otimes F ) $}.
\end{lemma}
\begin{proof}
\mbox{$ F . ( G \cap H )  \subseteq F . G $} and
\mbox{$ { {( F . ( G \cap H ) )}^\bot }^\bot \subseteq $}
\mbox{$ {{( F . G )}^\bot}^\bot = $} \mbox{$ F \otimes G $}.
The reader will easily complete the proof.
\QED \end{proof}

\subsection{{\em plus}, the additive disjunction} \label{sec:plus}
\begin{definition} \label{def:plus}
If \mbox{$F, G \subseteq {\cal P} $} are facts, 
\mbox{$ F \oplus G  = $} \mbox{$ {{( F \cup G )}^\bot}^\bot $}.
\end{definition}
By Lemma~\ref{the:fact}, \mbox{$ F \oplus G $} is a fact.

One easily sees that \mbox{$ F \oplus G = $} \mbox{$ \mysim ( \mysim F \& \mysim G ) $}
and \mbox{$ F \& G = $} \mbox{$ \mysim ( \mysim F \oplus \mysim G ) $},
that {\em plus} is associative and commutative, that {\em times} distributes over {\em plus},
that ${\cal P}$ is a zero element for {\em plus} and that {\em par} semi-distributes 
over {\em plus}: i.e., \mbox{$ ( F \dna G ) \oplus ( F \dna H ) \subseteq $}
\mbox{$ F \dna ( G \oplus H ) $}.

\section{Sequents} \label{sec:sequents}
We consider the language proposed by Girard on p. 21 of~\cite{Girard:87}, where negation 
can only be applied to atomic propositions, but limit the individual constants to $\bf 1$ and
$\top$ since $\bot$ and $\bf 0$ can be defined as their linear negations respectively.
Given a phase space, if every propositional variable is assigned a {\em fact} every
formula defines a fact.
The constant {\bf 1} denotes the set ${\cal Z}^\bot$ and $\top$ denotes $\cal P$.

Since our {\em times} and {\em par} connectives are not associative, we must decide
how the sequences are interpreted: we choose association to the left.
This is consistent with the remark after Definition 1 in~\cite{Lehmann_andthen:JLC} that
\mbox{$ ( F . G ) . H $} has an immediate interpretation: first $F$, then $G$, finally $H$,
whereas \mbox{$ F . ( G . H ) $} does not.
\begin{definition} \label{def:sequent}
A sequent
\[
A_{1} , A_{2} , \ldots , A_{n} \vdash B_{1} , B_{2} , \ldots , B_{m}
\]
where the A's and the B's are facts is valid in a Q-structure iff the fact
\[
( ( \ldots ( A_{1} \otimes A_{2} ) \otimes \ldots ) \otimes A_{n} ) \multimap 
( ( \ldots ( B_{1} \dna B_{2} ) \dna \ldots \dna ) B_{m} )
\]
is valid in the structure.
A sequent is {\em valid} iff it is valid in any Q-structure.
\end{definition}

\section{Proof rules for multiplicative and additive connectives: soundness} 
\label{sec:proof_rules}
We shall present sound proof rules for the logic of Q-structures.
As in the commutative case, by Lemma~\ref{the:passing}, we can consider only sequents
whose left-and side is empty.
Note that the sequent \mbox{$ \ \vdash A_{1} , \ldots , A_{n} $} is valid iff 
\mbox{$1 \in ( \ldots ( A_{1} \dna A_{2} ) \ldots ) \dna A_{n} $}.

\subsection{Logical axioms} \label{sec:logical_axioms}
A logical axiom:
\mbox{$\ \vdash \mysim A , A$}.
\newline Soundness: \mbox{$ A \dna A^\bot =$} \mbox{$ {( A^\bot . A )}^\bot $}
and \mbox{$ A^\bot . A \subseteq {\cal Z} $}.
Therefore \mbox{$ {\cal Z}^\bot \subseteq$} \mbox{$ A \dna A^\bot $}.
But \mbox{$ 1 \in { \cal Z }^\bot $}.

\subsection{Cut rule} \label{sec:cut_rule}
\[ \begin{array}{c}
\begin{array}{c}
\vdash A , B  \ \ \vdash \mysim A , C  \\
\hline 
\vdash B , C 
\end{array} 
\end{array} \]
Soundness: one easily sees, for example by Corollary~\ref{the:1par}, that
\mbox{$ 1 \in A \dna B $} iff \mbox{$A^\bot \subseteq B $} and
\mbox{$ 1 \in \mysim A \dna C $} iff \mbox{$ A \subseteq C $}.
If both assumptions hold, then we have \mbox{$ A^\bot . A \subseteq B . C $} and
\mbox{$ A^\bot \dna A \subseteq B \dna C $}.
But we have seen in subsection~\ref{sec:logical_axioms} that \mbox{$ 1 \in A^\bot \dna A $}.

\subsection{Exchange rules} \label{sec:exchange_rule}
There is no sweeping exchange rule as in~\cite{Girard:87} 
but there are two limited exchange rules.
\[ {\bf Exchange1}
\begin{array}{c}
\vdash A_{1} , A_{2} \\
\hline
\vdash A_{2} , A_{1}
\end{array} \]
Soundness: by Corollary~\ref{the:1par}, \mbox{$ 1 \in A_{1} \dna A_{2} $} iff
\mbox{$ A_{1}^\bot \subseteq A_{2} $} iff \mbox{$ A_{2}^\bot \subseteq A_{1} $} iff
\mbox{$ 1 \in B \dna A $}.

\[ {\bf Exchange2} \ \ \ \ 
\begin{array}{c}
\vdash A_{1} , A_{2} , A_{3} \\
\hline
\vdash A_{3} , A_{2} , A_{1}
\end{array} \]
Let us fix a Q-structure.
Assume \mbox{$ \vdash A_{1} , A_{2} , A_{3} $} is valid in the Q-structure.
By Lemma~\ref{the:passing} 
\mbox{$ \mysim A_{3} \vdash A_{1} , A_{2} $} is valid.
By Lemma~\ref{the:implication}, 
\mbox{$ A_{3}^{\bot} \subseteq A_{1} \dna A_{2} $}.
Therefore \mbox{$ {( A_{1} \dna A_{2} )}^\bot \subseteq A_{3} $} and
\mbox{$  A_{1}^\bot \otimes A_{2}^\bot \subseteq A_{3} $} and
\mbox{$ \mysim A_{1} \otimes \mysim A_{2} \vdash A_{3} $}.
By Lemma~\ref{the:passing} again:
\mbox{$ \mysim A_{1} \vdash A_{3} , A_{2} $} and
\mbox{$ \vdash A_{3} , A_{2} , A_{1} $} in the Q-structure.

\subsection{Additive rules} \label{sec:additive_rules}
\begin{itemize}

\item 
\mbox{$ \vdash \top , A \ \ \ {\rm \bf Axiom } \top $}.
\newline Soundness: \mbox{$ {\cal P} \dna A =$} \mbox{$ {\bf 0} \multimap A $} and,
by Lemma~\ref{the:fact}, \mbox{${\bf 0} \subseteq A$} and 
\mbox{$ 1 \in {\bf 0} \multimap A $}.

\item 
\[ \begin{array}{cc}
\begin{array}{c}
\vdash A , C \ \  \vdash B , C  \\
\hline 
\vdash A \& B , C 
\end{array} 
\end{array} \ \ \  {\bf \&} \]
Soundness follows from the distributivity of {\em par} over {\em with}:
\mbox{$ ( A \& B ) \dna C =$} \mbox{$ ( A \dna C ) \& ( B \dna C ) $}
and therefore, if \mbox{$1 \in A \dna C $} and \mbox{$1 \in B \dna C $}, we have
\mbox{$ 1 \in ( A \& B ) \dna C $}.

\item 
\[ \begin{array}{cc}
\begin{array}{c}
\vdash A ,C \ \ \\
\hline
\vdash A \oplus B , C
\end{array} \ \ \ {\bf \oplus}1 \ \ \ \ \ \ \ &
\begin{array}{c}
\vdash A , C \ \  \\
\hline 
\vdash B \oplus A , C 
\end{array} \ \ \ \ {\bf \oplus}2
\end{array} \]
Soundness of ${\oplus}1$ follows from the half-distributivity of {\em par} over {\em plus}:
\mbox{$ ( A \dna C ) \oplus ( B \dna C ) \subseteq $}
\mbox{$ ( A \oplus B ) \dna C $} and therefore \mbox{$ 1 \in A \dna C $} implies
\mbox{$ 1 \in ( A \oplus B ) \dna C $}.
\end{itemize}

\subsection{Multiplicative rules} \label{sec:multiplicative rules}
\begin{itemize}
\item 
\mbox{$ \vdash {\bf 1 } \ \ \ \ {\rm \bf Axiom1} $}
\newline Soundness: \mbox{$ 1 \in {\{ 1 \}^\bot}^\bot $}.

\item 
\[ \begin{array} {c}
\vdash A \\
\hline 
\vdash \, \mysim {\bf 1} , A 
\end{array} \ \ \  {\bf \bot} \]
Soundness follows from \mbox{$ {\cal Z} \dna A = $} $ A $.

\item 
\[ \begin{array}{c}
\vdash A , C \ \ \ \  \ \ \  \vdash B , D \\
\hline 
\vdash C , D , A  \otimes B 
\end{array} \ \ \  {\bf \otimes} \]
Soundness is proved by: 
\newline The assumptions are equivalent to \mbox{$ C^\bot \subseteq A $}
and \mbox{$ D^\bot \subseteq B $}.
Therefore \mbox{$ C^\bot . D^\bot \subseteq A . B $}
and \mbox{$ {( C \dna D )}^\bot \subseteq A \otimes B $}.
We see, by Lemma~\ref{the:implication} that
\mbox{$ 1 \in {( C \dna D )}^\bot \multimap ( A \otimes B ) $}.
We conclude that \mbox{$ 1 \in ( C \dna D ) \dna ( A \otimes B ) $}.

\item 
\[ \begin{array}{c}
\vdash A , B , \sigma  \\
\hline 
\vdash A \dna B , \sigma 
\end{array}  \ \ \ {\bf \dna} \]
where $\sigma$ is any sequence of formulas.
Soundness follows from our interpretation of the commas in a sequent as a left-associative
{\em par} connective.
\end{itemize}

It is now clear that a sequent that is provable from the rules described
above is valid in any Q-structure.
\begin{theorem}[Soundness] \label{the:soundness}
Any sequent provable from the axioms and the rules above is valid in any Q-structure.
\end{theorem}
\begin{proof}
By induction on the length of the proof.
\QED \end{proof}

\section{Completeness} \label{sec:completeness}
The proof that the rules above are complete for Q-structures follows the line of the 
corresponding proof in~\cite{Girard:87}, but the differences require attention.
Since the original exchange rule has been replaced by much weaker rules, the side of the 
sequents must be considered as sequences and not as multi-sets.
Since the connectives $\dna$ and $\otimes$ are not associative, the elements
of the universal phase structure central in the completeness proof cannot be sequences and
concatenation, they have to be formulas and the composition must be $\dna$.

We shall define a suitable Q-structure.
Let $\cal L$ be the propositional language defined in Section~\ref{sec:sequents}.
The carrier of our Q-structure, \mbox{$ M \eqdef {\cal L} \cup \{ \epsilon \} $} 
contains the propositions and a distinguished element $\epsilon$.

The ``.'' operation is defined by
\mbox{$ x . y \eqdef $} \mbox{$ x \dna y $} if \mbox{$ x , y \in {\cal L} $} and
\mbox{$ x . \epsilon = x . \epsilon = x $} for any \mbox{$ x \in M $}.
The distinguished element $\epsilon$ is a neutral element for ``.''.

To define the garbage set $\cal Z$ we need some notation.
Let \mbox{$\sigma \in {\cal L}^\ast$} be a sequence of propositions: 
\mbox{$ \sigma = A_{1} , \ldots , A_{n} $} with \mbox{$ n \geq 2 $}.
The sequence $\sigma$ defines a proposition 
\mbox{$\bar{\sigma} =$} \mbox{$ ( \ldots ( A_{1} \dna A_{2} ) \ldots ) \dna A_{n} $} 
where the propositions are connected with the {\em par} connective $\dna$ in 
a left-associative way.
We extend the definition by setting \mbox{$ \bar{A} = A $} for any \mbox{$ A \in \cal L $}
and \mbox{$\bar{\epsilon} = \mysim {\bf 1} $}.
The garbage set $\cal Z$ can now be defined by: 
\mbox{$ {\cal Z } = $} 
\mbox{$ \{ \bar{\sigma} \mid \sigma \in {\cal L}^\ast {\rm \ such \ that \ } 
\vdash \sigma \} $}.

We shall now verify that the conditions~\ref{symmetry} and~\ref{reverse} 
of Definition~\ref{def:Q-structure} are satisfied.
Cases involving $\epsilon$ are easily treated and we may assume \mbox{$x , y \in \cal L $}.
For~\ref{symmetry}, note that 
\mbox{$ x . y \in {\cal Z} $} iff \mbox{$ x \dna y \in {\cal Z} $} 
iff \mbox{$ \vdash x , y $} iff,
by our Exchange1 rule, \mbox{$ \vdash y , x $} iff
\mbox{$ y . x \in {\cal Z} $}.

For~\ref{reverse}, 
\mbox{$ ( x . y ) . z \in {\cal Z} $} iff
\mbox{$ ( x \dna y ) \dna z \in \cal Z $} iff
\mbox{$ \vdash x , y , z $}.
By {\bf Exchange2} this is equivalent to 
\mbox{$ \vdash z , y , x $} and to
\mbox{$ \vdash ( z \dna y ) \dna x $} which is equivalent to
\mbox{$ ( z . y ) . x \in \cal Z $} and to
\mbox{$ x . ( z . y ) \in \cal Z $} by~\ref{symmetry}.

We have just defined a Q-structure, that we shall call $M$, as its carrier. 
Note that the definition of $\cal Z$ implies that for any \mbox{$ x , y \in \cal L $},
\mbox{$ x \bot y $} iff \mbox{$ \vdash x , y $}.

To any formula \mbox{$ x \in \cal L $} we shall associate a subset of $M$, $S(x)$. 
We intend $S(x)$ to be a fact in the Q-structure $M$ for any $x$.
The definition of $S(x)$ proceeds by induction on the size of $x$.
\begin{enumerate}
\item
\mbox{$ S ( {\bf 1} ) = $} \mbox{$ {\cal Z}^\bot $}, 
\item
\mbox{$ S ( \top) = M $},
\item
\mbox{$ S ( \mysim x ) = $} \mbox{$ ( S(x) )^\bot$}, 
\item
\mbox{$ S ( x \& y ) = $} \mbox{$ S(x) \cap S(y) $},
\item
\mbox{$ S ( x \oplus y ) = $} \mbox{$ { {( S ( x ) \cup S ( y ) )}^\bot}^\bot $},
\item
\mbox{$ S ( x \otimes y ) = $} \mbox{$ { {( S ( x ) .  S ( y ) )}^\bot}^\bot $},
\item
\mbox{$ S ( x \dna y ) = $} \mbox{$ {( {S ( x )}^\bot . {S( y )}^\bot )}^\bot $},
\item
for every propositional letter $a$, \mbox{$ S ( a ) = $}
\mbox{$ Pr ( a ) $} as defined in Definition~\ref{def:Pr} just below.
\end{enumerate}

\begin{definition} \label{def:Pr}
For any \mbox{$x \in \cal L $}, we let 
\[
Pr( x ) \eqdef \{ x \}^\bot.
\]
For every \mbox{$x \in \cal L$}, $Pr(x)$ is a fact of the Q-structure $M$.
\end{definition}
An equivalent definition is:
\[
Pr( x ) = 
\{ \bar{\sigma} \mid \sigma \in {\cal L}^\ast {\rm \ such \ that \ } \vdash \sigma , x \}.
\]

\begin{lemma} \label{the:fact2}
For any formula $x$, \mbox{$ {( Pr(x) )}^\bot = Pr( \mysim x ) $}.
\end{lemma}
\begin{proof}
Let \mbox{$y \in Pr(x) $} and \mbox{$ z \in Pr(\mysim x) $}.
We have \mbox{$ y \in \{ x \}^\bot $} and \mbox{$ z \in \{ \mysim x \}^\bot $}.
Therefore \mbox{$ \vdash x , y $} and \mbox{$ \vdash \mysim x , z $}.
By Cut we conclude that \mbox{$ \vdash y , z $}, and therefore,
\mbox{$ z \in \{ y \}^\bot $}.
We have shown that \mbox{$ Pr ( \mysim x ) \subseteq $} \mbox{$ { Pr ( x ) }^\bot $}.

Conversely, since \mbox{$ \vdash \, \mysim x , x $} is an axiom 
\mbox{$ \mysim x \in \{x\}^\bot =$} \mbox{$ Pr(x) $} and 
\mbox{$ { Pr(x) }^\bot \subseteq {( \mysim x )}^\bot =$} \mbox{$ Pr ( \mysim x ) $}.
\QED \end{proof}

We want, now, to show that the interpretation $S ( x )$ in the Q-structure we defined 
for any formula $x$ is exactly $Pr(x)$.
\begin{lemma} \label{the:Pr}
For any \mbox{$ x \in \cal L $}, \mbox{$ S ( x ) = Pr ( x ) $}.
\end{lemma}
\begin{proof}
By induction on the length of the formula $x$.
\begin{enumerate}
\item
Let \mbox{$ x = a $} for a propositional letter $a$. By construction we have
\mbox{$ S(x) = Pr(x)  $}.
\item
Let \mbox{$ x = {\bf 1}$}, by Axiom1 we have \mbox{$ {\bf 1} \in \cal Z $},
\mbox{$ \{ {\bf 1} \} \subseteq \cal Z $}, \mbox{$ {\cal Z}^\bot \subseteq Pr ( {\bf 1} ) $}. 
Assume, now, that \mbox{$y \in Pr({\bf 1})$}.
We have \mbox{$ \vdash {\bf 1} , y $}.
Let \mbox{$z \in \cal Z $}. We have \mbox{$ \vdash z $}. 
By Rule $\bot$ we have \mbox{$ \bot , z $}, i.e., \mbox{$ \vdash \mysim {\bf 1} , z $} 
and by the Cut rule we have \mbox{$ \vdash z , y $} and \mbox{$ y \in {\cal Z}^\bot $}.
We have now shown that \mbox{$ Pr({\bf 1}) \subseteq {\cal Z}^\bot $}.
\item
Let \mbox{$ x = \top$}. By Axiom$\top$, \mbox{$ Pr(\top) = M $}.
\item
Let \mbox{$ x = \mysim y $}. We have\mbox{$ S ( \mysim x ) = $}
\mbox{$ {S(x)}^\bot =$}  \mbox{$ {Pr(x)}^\bot = $} \mbox{$ Pr ( \mysim x ) $}
by Lemma~\ref{the:fact2}.
\item
Let \mbox{$ x = y \& z $}. We have \mbox{$ S ( y \& z ) =$} \mbox{$ S ( y ) \cap S(z) =$}
\mbox{$ Pr ( y ) \cap Pr ( z ) \subseteq $} \mbox{$ Pr ( y \& z ) $} by Rule\&.
For the converse inclusion, the proof is the classical one, see~\cite{Girard:87} for example:
one only needs to check that Exchange is not used.
\item
Let \mbox{$ x = y \oplus z $}. By duality with the previous case.
\item
Let \mbox{$ x = y \otimes z $}. 
We have \mbox{$ S ( y \otimes z ) = $} \mbox{$ { {( S ( y ) . S ( z ) )}^\bot}^\bot =$}
\mbox{$ { {( Pr ( y ) . Pr ( z ) )}^\bot}^\bot $}.

For any \mbox{$ u \in Pr(y) $}, \mbox{$ t \in Pr(z) $}, we have
\mbox{$ \vdash y , u $} and \mbox{$ \vdash z , t $}.
By Rule $\otimes$, then, we have \mbox{$ \vdash u , t , y \otimes z $},
\mbox{$ \vdash u \dna t , y \otimes z $}, by Exchange1 
\mbox{$ \vdash y \otimes z , u \dna t $} and 
\mbox{$ u \dna t \in Pr( y \otimes z ) $}.
We have proved that \mbox{$ Pr(y) \dna Pr(z) \subseteq Pr(y \otimes z) $}.
But \mbox{$ Pr(y) . Pr(z) =$} \mbox{$ Pr(y) \dna Pr(z) $}.
Therefore \mbox{$ { {( Pr(y) . Pr(z) )}^\bot}^\bot \subseteq$}
\mbox{$ { Pr( y \otimes z)^\bot}^\bot = Pr( y \otimes z) $}.
We have proved that \mbox{$ S( y \otimes z ) \subseteq $} \mbox{$ Pr(y \otimes z) $}.

Conversely, let \mbox{$ t \in {( Pr ( y ) . Pr ( z ) )}^\bot $}.
For any \mbox{$ u \in Pr(y) $}, \mbox{$ v \in Pr(z) $} we have
\mbox{$ \vdash u \dna v , t $} and therefore \mbox{$ \vdash Pr(y) \dna Pr(z) , t $}.
We see that \mbox{$ \vdash \mysim y \dna \mysim z ,  t $} and
\mbox{$ \vdash \mysim ( y \otimes z ) , t $}.
For any \mbox{$ w \in Pr ( y \otimes z ) $}, we have \mbox{$ \vdash y \otimes z , w $}.
By Cut we get \mbox{$ \vdash t , w $} and we conclude that 
\mbox{$ Pr ( y \otimes z ) \subseteq $} \mbox{$ Pr ( y ) . Pr ( z ) \subseteq $}
\mbox{$ { {( Pr ( y ) . Pr ( z ) )}^\bot}^\bot =$} \mbox{$ S ( y \otimes z ) $}.
\item
Let \mbox{$ x = y \dna z $}. By duality with the previous case.
\end{enumerate}
\QED \end{proof}

We can now conclude.
\begin{theorem} \label{the:completeness}
Any formula valid in any Q-structure and any assignment of facts to propositional letters 
is provable in the system of axioms and rules presented in Section~\ref{sec:proof_rules}.
\end{theorem}
\begin{proof}
Any formula $x$ valid in the Q-structure $M$ satisfies \mbox{$ \epsilon \in S ( x ) $}
and therefore, by Lemma~\ref{the:Pr} \mbox{$ \vdash x , \epsilon $}.
But \mbox{$ \epsilon^\bot = {\bf 1} $} and, by Axiom1 and Cut we get
\mbox{$ \vdash x $}.
\QED \end{proof}

The present effort does not propose exponential connectives for Q-structures because no sound
and complete rules were found for the natural generalization of topolinear structures to
non-associative phase structures. 
In particular, the formula \mbox{$ ? A \dna ? B \multimapboth ? ( A \oplus B ) $}, central in
Girard's treatment, is not valid in such structures: Girard's use of the Exchange rule cannot
be circumvented by Exchange1 and Exchange2 .

\section{Projective structures} \label{sec:projective}
The Q-structures we have described present a sub-structural logic without Contraction or
Weakening and with a very limited Exchange rule.
It generalizes the Linear Logic of Girard's~\cite{Girard:87} and is satisfied 
by the baby quantum logic of Section~\ref{sec:baby}.
But some additional rules seem to be valid in Quantum Logic.
We shall study limited forms of Contraction and Weakening.

If one performs a measurement on a quantic system, it is a fundamental principle that
a second performance of the same measurement will not change the state of the system:
the same value will be obtained with probability one.
So it seems that the rule:
\[ \begin{array}{c}
\vdash \sigma , A , \tau \\
\hline
\vdash \sigma , A , A , \tau
\end{array} \]
should be valid.
This is a limited form of Weakening.

The inverse Contraction rule also seems to be valid.

If we understand the sequent 
\mbox{$ A_{1} , \ldots , A_{n} \vdash B_{1} , \ldots , B_{m} $} as meaning that
any state resulting from the sequence of measurements on the left satisfies the condition
described by the sequence on the right, we would expect that any extension on the left of
the left-hand side can only restrict the set of final states and therefore we expect the rule
\[ \begin{array}{c}
\sigma \vdash \tau  \\
\hline
A , \sigma \vdash \tau 
\end{array} \]
should be valid.
This is a limited form of Weakening.

Note that Baby Quantum Logic and our presentation of classical logic satisfy the rules above.

\begin{definition} \label{def:projective}
A Q-structure \mbox{$\langle {\cal P} , {\cal Z} , . , 1 \rangle $} is {\em projective} iff
the operation ``.'' absorbs into ${\cal Z}$:
for any \mbox{$ x \in {\cal Z} $} and any \mbox{$ y \in {\cal P} $}
one has \mbox{$ x . y \in {\cal Z} $}.
\end{definition}

\begin{lemma} \label{the:xx=x}
In a projective Q-structure, for any \mbox{$ x , y , z \in {\cal P} $}:
\begin{enumerate}
\item \label{idempotence}
\mbox{$ x \bot y $} iff \mbox{$ x . x \bot y $},
\item \label{ortho}
if \mbox{$ x \bot z $} and \mbox{$ y \bot z $}, then \mbox{$ x . y \bot z $}.
\end{enumerate}
\end{lemma}
Property~\ref{idempotence} essentially means \mbox{$ x . x = x $},
and hence the term {\em projective}.
It expresses the idea that combining something with itself leaves the situation 
essentially unchanged.
Property~\ref{ortho} expresses the idea that if both $x$ and $y$ are incompatible with $z$, 
the combination $ x . y $ must also be incompatible with $z$. 
Note that our baby quantum logic is projective, and so is our presentation of classical logic in
Section~\ref{sec:Q-structures}.

\begin{proof}
\begin{enumerate}
\item
\mbox{$ x \bot y $} iff \mbox{$ y . x \in {\cal Z} $} implies \mbox{$ ( y . x) . x \in {\cal Z} $}
since the structure is projective.
But \mbox{$ ( y . x) . x \in {\cal Z} $} implies \mbox{$ y . ( x . x ) \in {\cal Z} $}
by Definition~\ref{def:Q-structure}.
\item
\mbox{$ z . y \in {\cal Z} $} implies \mbox{$ x . ( z . y ) \in {\cal Z} $} which implies
\mbox{$ ( x . y ) . z \in {\cal Z} $}.
\end{enumerate}
\QED \end{proof}

Our next result presents basic properties of projective Q-structures.
\begin{lemma} \label{the:projective}
In a projective Q-structure:
\begin{enumerate}
\item \label{sub}
for any \mbox{$ A \subseteq {\cal P} $}, \mbox{$ A \subseteq {( A . A )^\bot}^\bot $},
\item \label{subbot}
for any \mbox{$ A \subseteq {\cal P} $}, \mbox{$ A^\bot . A^\bot \subseteq A^\bot $},
\item \label{equal}
for any fact $F$, \mbox{$ F . F \subseteq F $}, \mbox{$ F \otimes F = F $} and 
\mbox{$ F \dna F = F $},
\item \label{1P}
\mbox{$ {\bf 1} = {\cal P} $} and \mbox{$ {\cal Z} = {\bf 0} $},
\item \label{ZA}
for any fact $F$, one has \mbox{$ {\cal Z} \subseteq F $}.
\end{enumerate}
\end{lemma}
\begin{proof}
\begin{enumerate}
\item
By property~\ref{idempotence} of Lemma~\ref{the:xx=x}:
for any \mbox{$ a \in A $}, \mbox{$ a \in { \{ a \}^\bot}^\bot $} and
\mbox{$ { a . a \in \{ a \}^\bot}^\bot $}.
We see that \mbox{$ A \subseteq {( A . A )^\bot}^\bot $}.
\item
By property~\ref{ortho} in Lemma~\ref{the:xx=x}.
\item
By item~\ref{subbot} \mbox{$ {F^\bot}^\bot . {F^\bot}^\bot \subseteq {F^\bot}^\bot $}
and therefore \mbox{$ F . F \subseteq F $}.
We see that \mbox{$ F \otimes F \subseteq $} \mbox{$ {F^\bot}^\bot = F $}.
By item~\ref{sub} \mbox{$ F \subseteq F \otimes F $}.
By duality, one easily sees that \mbox{$ F \dna F = F $}.
\item
By the absorption property \mbox{$ {\cal Z}^\bot = {\cal P} $}.
\item
For any \mbox{$ A \subseteq {\cal P } $}, by the absorption property, one has
\mbox{$ {\cal Z} \subseteq A^{\bot} $}.
\end{enumerate}
\QED \end{proof}

\begin{lemma} \label{the:main_proj}
In a projective Q-structure, for any \mbox{$ A , B \subseteq {\cal P} $},
one has \mbox{$ A^\bot \subseteq ( B . A )^\bot $} and 
\mbox{$ A \subseteq ( B . A^\bot)^\bot $}.
\end{lemma}
\begin{proof}
Let \mbox{$ x \in A^\bot , y \in A , z \in B $}.
We have \mbox{$ x . y \in {\cal Z} $} and therefore, by absorption, 
\mbox{$ ( x . y ) . z \in {\cal Z} $} and \mbox{$ x . ( z . y ) \in {\cal Z} $}.
We conclude that \mbox{$ A^\bot \subseteq ( B . A )^\bot $}.
Similarly \mbox{$ y . x \in {\cal Z} $}, \mbox{$ ( y . x ) . z \in {\cal Z} $},
\mbox{$ y . ( z . x ) \in {\cal Z} $} and \mbox{$ A \subseteq ( B . A^\bot )^\bot $}.
\QED \end{proof}

\begin{corollary} \label{the:notimes}
For any \mbox{$ A \subseteq {\cal P} $} and any facts $F$ , $G$: 
\mbox{$ A . F \subseteq F $}, \mbox{$ G \otimes F \subseteq F $} and
\mbox{$ F \subseteq G \dna F $}.
\end{corollary}
\begin{proof}
By lemma~\ref{the:main_proj}, \mbox{$ F^\bot \subseteq ( A . F )^\bot $}
and therefore \mbox{$ { ( A . F )^\bot}^\bot \subseteq F $} and 
\mbox{$ G \otimes F \subseteq F $}.
Last claim is proved by duality.
\QED \end{proof}

Our last result shows that the Contraction rule discussed above is valid.
\begin{lemma} \label{the:contraction}
In a projective Q-structure, for any facts $F$, $G$, one has 
\mbox{$ ( G \dna F ) \dna F \subseteq$} \mbox{$ G \dna F $}.
\end{lemma}
\begin{proof}
By Corollary~\ref{the:notimes} \mbox{$ F \subseteq G \dna F $} and therefore
\mbox{$ ( G \dna F ) \dna F \subseteq $} \mbox{$ ( G \dna F ) \dna ( G \dna F )$}.
Lemma~\ref{the:projective}, then, shows that 
\mbox{$ ( G \dna F ) \dna F \subseteq $} \mbox{$ G \dna F $}.
\QED \end{proof}

\section{Rule for the projective case} \label{sec:rules_projective}

The added rule, that expresses the essence of our requirement that 
$\cal Z$ be absorbing is a Weakening rule.
\[ \begin{array}{c}
\vdash A   \\
\hline
\vdash \ A , B
\end{array} {\bf \ \ WR \ \ } (Right-Weakening) \]
Soundness follows from the following.
\begin{lemma} \label{the:WR}
For any facts $A$, $B$, if \mbox{$ {\bf 1} \subseteq A $}, then
\mbox{$ {\bf 1} \subseteq A \dna B $}.
\end{lemma}
\begin{proof}
\mbox{$ {\bf 1} \subseteq A $} is equivalent to \mbox{$ A^\bot \subseteq {\cal Z} $},
which implies \mbox{$ A^\bot . B^\bot \subseteq$} 
\mbox{$ {\cal Z} . B^\bot \subseteq {\cal Z} $}
by the absorption property of Definition~\ref{def:projective}.
We conclude that \mbox{$ {\bf 1} =$} \mbox{$ {\cal Z}^\bot \subseteq$}
\mbox{$ A \dna B $}.
\QED \end{proof}

\begin{theorem} \label{the:exp}
The system of the eleven rules and axioms in Section~\ref{sec:proof_rules} and the added
WR rule is sound and complete for projective Q-structures phase semantics.
\end{theorem}
\begin{proof}
Soundness has been proved on the way.
For completeness, we shall use the technique used in Section~\ref{sec:completeness}.

Our only task is to show that, under the new rules, the Q-structure built is projective.
This is guaranteed by the Right-Weakening rule WR.
We want to show that, if \mbox{$ A , B \in {\cal Z} $}, one has \mbox{$ A . B \in {\cal Z} $},
i.e., \mbox{$ A \dna B \in {\cal Z} $}.
Assume \mbox{$ \vdash A $}, by WR we have \mbox{$ \vdash A , B $}, 
\mbox{$ \vdash A\dna B $} and \mbox{$ A . B \in {\cal Z} $}.
\end{proof}

We can now conclude.
\begin{theorem} \label{the:completeness}
Any formula valid in any projective Q-structure and 
any assignment of facts to propositional letters is provable in the system of axioms and rules 
presented in Sections~\ref{sec:proof_rules} and~\ref{sec:rules_projective}.
\end{theorem}
\begin{proof}
Any formula $x$ valid in any projective Q-structure $M$ satisfies 
\mbox{$ \epsilon \in S ( x ) $}
and therefore, by Lemma~\ref{the:exp} \mbox{$ \vdash x , \epsilon $}.
But \mbox{$ \epsilon^\bot = {\bf 1} $} and, by Axiom1 and Cut we get
\mbox{$ \vdash x $}.
\QED \end{proof}

\bibliographystyle{plain}

\end{document}